\documentclass{article}

\usepackage[latin1]{inputenc}
\usepackage[centertags]{amsmath}
\usepackage{amsfonts}
\usepackage{amssymb}
\usepackage{amsthm}
\usepackage{newlfont}
\usepackage{mathrsfs}
\usepackage{amsbsy}

\newcommand{\nat}{\mathbb N}
\newcommand{\real}{\mathbb R}
\newcommand{\F}{\mathcal{F}}
\newcommand{\M}{\mathcal{M}}

\newcommand{\Q}{\mathcal{Q}}
\newcommand{\A}{\mathcal{A}}
\newcommand{\Y}{\mathcal{Y}}
\newcommand{\Pc}{\mathcal{P}}
\newcommand{\T}{\mathbb{T}}
\newcommand{\Qs}{{\mathcal Q}^{\ast}}
\newcommand{\B}{{\mathcal B}}
\newcommand{\g}{\gamma}

\newcommand{\ma}{\M_1(P)}
\newcommand{\me}{\M^e(P)}

\newcommand{\LT}{L^{\infty}}
\newcommand{\Lt}{L^{\infty}_t}
\newcommand{\LTq}{L^{\infty}(\Omega, \F,P)}

\newcommand{\lk}{\left\{\,}
\newcommand{\rk}{\right\}}
\newcommand{\mk}{\;\big|\;}

\newcommand{\lp}{\left[}
\newcommand{\rp}{\right]}

\DeclareMathOperator*{\es}{ess\,sup}
\DeclareMathOperator*{\ei}{ess\,inf}
\DeclareMathOperator*{\qes}{Q\text{-}ess\,sup}

\DeclareMathOperator*{\res}{R\text{-}ess\,sup}
\newcommand{\esu}{\es_{Q\in\Q_t}}

\newcommand{\pk}{,\ldots ,}

\newcommand{\zt}{_{t\in\T}}

\newcommand{\At}{{\mathcal A}_{t+1}}
\newcommand{\Att}{{\mathcal A}_{t,t+1}}

\newcommand{\pfm}{\alpha_t^{\min}(Q)}
\newcommand{\pf}{\alpha_t^{\min}(Q)}
\newcommand{\pfo}{\alpha_t^{\min}}
\newcommand{\pfn}{\alpha_0^{\min}(Q)}
\newcommand{\pft}{\alpha_{t+1}^{\min}(Q)}
\newcommand{\pfs}{\alpha_{t+s}^{\min}(Q)}
\newcommand{\pftt}{\alpha_{t, t+1}^{\min}(Q)}
\newcommand{\pfts}{\alpha_{t, t+s}^{\min}(Q)}
\newcommand{\pfk}{\alpha_{k,k+1}^{\min}(Q)}

\newcommand{\rt}{\rho_t}

\newcommand{\rts}{(\rt)\zt}

\newcommand{\rtt}{\rho_{t+1}}
\newcommand{\rs}{\rho_{t+s}}

\newcommand{\tq}{\widetilde{Q}}

\newcommand{\ew}{E_Q[-X\,|\,\F_t\,]}
\newcommand{\ewo}{[-X\,|\,\F_t\,]}
\newcommand{\f}{P\mbox{-a.s.}}
\newcommand{\qf}{Q\mbox{-a.s.}}
\newcommand{\ra}{\Rightarrow}

\newcommand{\crm}{conditional convex risk measure }

\newcommand{\cts}{continuous from above }
\newcommand{\ctse}{continuous from above}

\newtheorem{theorem}{Theorem}

\newtheorem{corollary}[theorem]{Corollary}
\newtheorem{definition}[theorem]{Definition}
\newtheorem{lemma}[theorem]{Lemma}
\newtheorem{proposition}[theorem]{Proposition}
\newtheorem{remark}[theorem]{Remark}
\newtheorem{example}[theorem]{Example}

\title{Dynamic risk measures}
\author{Beatrice Acciaio\thanks{Department of Economy, Finance and Statistics, University of Perugia, Via A. Pascoli 20, 06123 Perugia, Italy. Email: beatrice.acciaio@stat.unipg.it. Financial support from the European Science Foundation (ESF) ``Advanced Mathematical Methods for Finance" (AMaMeF) under the exchange grant 2281, and hospitality of Vienna University of Technology are gratefully acknowledged.} \and Irina Penner\thanks{Humboldt-Universit\"{a}t zu Berlin, Institut f\"{u}r Mathematik, Unter den Linden 6, 10099 Berlin, Germany. Email: penner@math.hu-berlin.de. Supported by the DFG Research Center \textsc{Matheon} ``Mathematics for key technologies''. Financial support from the European Science Foundation (ESF) ``Advanced Mathematical Methods for Finance" (AMaMeF) under the short visit grant 2854 is gratefully acknowledged.}}
\date{\small first version: November 13, 2009;  this version: \today}
\begin{document}
\maketitle

\begin{abstract}
This paper gives an overview of the theory of dynamic convex risk measures for random variables in discrete time setting. We summarize robust representation results of conditional convex risk measures, and we characterize various time consistency properties of dynamic risk measures in terms of acceptance sets, penalty functions, and by supermartingale properties of risk processes and penalty functions.
\end{abstract}

\section{Introduction}

Risk measures are quantitative tools developed to determine mimimum capital reserves, which are required to be maintained by financial institutions in order to ensure their financial stability.   
An axiomatic analysis of risk assessment in terms of capital requirements was initiated by  Artzner, Delbaen, Eber, and Heath~\cite{adeh97,adeh99}, who introduced coherent risk measures. F\"ollmer and Schied~\cite{fs2} and Frittelli and Rosazza Gianin~\cite{fr2} replaced positive homogeneity by convexity in the set of axioms and established the more general concept of a convex risk measure.   Since then, convex and coherent risk measures and their applications have attracted a growing interest both in mathematical finance research and among practitioners.

One of the most appealing properties of a convex risk measure is its robustness against model uncertainty. Under some regularity condition, it can be represented as a suitably modified worst expected loss over a whole class of probabilistic models. This was initially observed in \cite{adeh99, fs2, fr2} in the static setting, where financial positions are described by random variables on some probability space and a risk measure is a real-valued functional. For a comprehensive presentation of the theory of static coherent and convex risk measures we refer to Delbaen~\cite{d0} and F\"ollmer and Schied~\cite[Chapter 4]{fs4}.

A natural extension of a static risk measure is given by a conditional risk measure, which takes into account the information available at the time of risk assessment. As its static counterpart, a conditional convex risk measure can be represented as the worst conditional expected loss over a class of suitably penalized probability measures; see \cite{rse5,rie4,dt5,bn4,ks5, cdk6}. In the dynamical setting described by some filtered probability space, risk assessment is updated over the time in accordance with the new information. This leads to the notion of dynamic risk measure, which is a sequence of conditional risk measures adapted to the underlying filtration.

A crucial question in the dynamical framework is how risk evaluations at different times are interrelated.
Several notions of time consistency were introduced and studied in the literature. One of todays most used notions is strong time consistency, which corresponds to the dynamic programming principle; see \cite{adehk7,d6,dt5,ks5,cdk6,bn6,fp6,ck6,dpr10} and references therein. As shown in \cite{d6, bn6, fp6}, strong time consistency can be characterized by additivity of the acceptance sets and penalty functions, and also by a supermartingale property of the risk process and the penalty function process.  
Similar characterizations of the weaker notions of time consistency, so called rejection and acceptance consistency, were given in \cite{Samuel, ipen7}. Rejection consistency, also called prudence in \cite{ipen7}, seems to be a particularly suitable property from the point of view of a regulator, since it ensures that one always stays on the safe side when updating risk assessment. The weakest notions of time consistency considered in the literature are  weak acceptance and weak rejection consistency, which require that if some position is accepted (or rejected) for any scenario tomorrow, it should be already accepted (or rejected) today; see \cite{adehk7, Weber, tu8, burg, ros7}.

As pointed out in \cite{jr8, er08},  risk assessment in the multi-period setting should also account for uncertainty about the time value of money. This requires to consider entire cash flow processes rather than total amounts at terminal dates as risky objects, and it leads to a further extention of the notion of risk measure. Risk measures for processes were studied in \cite{adehk7, rie4, cdk4, cdk5, cdk6, ck6, fs6, jr8, afp9}. The new feature in this framework is that not only the amounts but also the timing of payments matters; cf.\ \cite{cdk6, ck6, jr8, afp9}. However, as shown in \cite{adehk7} in the static and in \cite{afp9} in the dynamical setting, risk measures for processes can be identified with risk measures for random variables on an appropriate product space. This allows a natural translation of results obtained in the framework of risk measures for random variables to the framework of processes; see \cite{afp9}.

The aim of this paper it to give an overview of the current theory of dynamic convex risk measures for random variables in discrete time setting; the corresponding results for risk measures for processes are given in \cite{afp9}.  The paper is organized as follows. Section~\ref{setup} recalls the definition of a conditional convex risk measure and its interpretation as the minimal capital requirement from \cite{dt5}. Section~\ref{sectionrr} summarizes robust representation results from \cite{dt5,fp6,bn8}. In Section~\ref{sec:tc} we first give an overview of different time consistency properties based on \cite{tu8}. Then we focus on the strong notion of time consistency, in Subsection~\ref{subsec:tc},  and we characterize it by supermartingale properties of risk processes and penalty functions. The results of this subsection  are mainly based on \cite{fp6}, with the difference that here we give characterizations of time consistency also in terms of absolutely continuous probability measures, similar to \cite{bn8}. In addition, we relate the martingale property of a risk process with the worst case measure, and we provide the explicit form of the Doob- and the Riesz-decomposition of the penalty function process. Subsection~\ref{subsec:rc} generalizes \cite[Sections 2.4, 2.5]{ipen7} and characterizes rejection and acceptance consistency in terms of acceptance sets, penalty functions, and, in case of rejection consistency, by a supermartingale property of risk processes and one-step penalty functions. Subsection~\ref{subsec:wc} recalls characterizations of weak time consistency from \cite{tu8, burg}, and Subsection~\ref{recur} characterizes the recursive construction of time consistent risk measures suggested in \cite{cdk6, ck6}. Finally, the dynamic entropic risk measure with a non-constant risk aversion parameter  is studied in Section~\ref{entropic}.

\section{Setup and notation}\label{setup}
Let $T\in\nat\cup\{\infty\}$ be the time horizon, $\T:=\{0,\ldots,T\}$ for $T<\infty$, and $\T:=\nat_0$ for $T=\infty$. We consider a discrete-time setting given by a filtered probability space
$(\Omega, \F, (\F_t)\zt, P)$ with $\F_0=\{\emptyset, \Omega\}$, $\F=\F_T$ for $T<\infty$, and $\displaystyle{\F=\sigma(\cup_{t\ge 0}\F_t)}$ for $T=\infty$. For $t\in\T$, $L^{\infty}_t:=L^{\infty}(\Omega, \F_t,P)$ is the space of all essentially bounded $\F_t$-measurable random variables, and $L^\infty:=L^{\infty}(\Omega, \F_T,P)$. All equalities and inequalities between random variables and between sets are understood to hold $P$-almost surely, unless stated otherwise. We denote by $\ma$ (resp. by $\me$) the set of all probability measures on $(\Omega, \F)$ which are absolutely continuous with respect to $P$ (resp. equivalent to $P$). 

In this work we consider risk measures defined on the set $L^{\infty}$, which is understood as the set of discounted terminal values of financial positions. In the dynamical setting, a conditional risk measure $\rt$ assigns to each terminal payoff $X$ an $\F_t$-measurable random variable $\rt(X)$, that quantifies the risk of the position $X$ given the information $\F_t$. A rigorous definition of a conditional convex risk measure was given in \cite[Definition 2]{dt5}.

\begin{definition}\label{defrm}
A map $\rt\,:\,\LT\,\rightarrow\,\Lt$ is called a \emph{conditional convex risk measure} if it satisfies the 
following properties for all $X,Y\in\LT$:
\begin{itemize}
\item[(i)]
Conditional cash invariance: For all $m_t\in\Lt$
\[\rt(X+m_t)=\rt(X)-m_t;\]
\item[(ii)]
Monotonicity: $X\le Y\;\,\Rightarrow\;\,\rt(X)\ge\rt(Y)$;
\item[(iii)]
Conditional convexity: for all $\lambda\in\Lt,\,0\le \lambda\le 1$:
\[
\rt(\lambda X+(1-\lambda)Y)\le\lambda\rt(X)+(1-\lambda)\rt(Y);
\]
\item[(iv)]
{Normalization}: $\rt(0)=0$.
\end{itemize}
A conditional convex risk measure is called a \emph{conditional coherent risk measure} if it has in addition 
the following property:
\begin{itemize}
\item[(iv)]
{Conditional positive homogeneity}: for all $\lambda\in\Lt,\,\lambda\ge0$:
\[
\rt(\lambda X)=\lambda\rt(X).
\] 
\end{itemize}
\end{definition}

In the dynamical framework one can also analyze risk assessment for cumulated cash flow \emph{processes} rather than just for terminal pay-offs, i.e.\ one can consider a risk measure that accounts not only for the amounts but also for the timing of payments.  Such risk measures were studied in \cite{cdk4, cdk5, cdk6, ck6, fs6, jr8, afp9}. As shown in \cite{adehk7} in the static and in \cite{afp9} in the dynamical setting, convex risk measures for processes can be identified with convex risk measures for random variables on an appropriate product space. This allows to extend results obtained in our present setting to the framework of processes; cf.\ \cite{afp9}.

If $\rt$ is a conditional convex risk measure, the function $\phi_t:=-\rho_t$ defines a conditional monetary utility function in the sense of \cite{cdk6, ck6}. 
The term ``monetary'' refers to conditional cash invariance of the utility function, the only property in Definition \ref{defrm} that does not come from the classical utility theory.  Conditional cash invariance is a natural request in view of the interpretation of $\rt$ as a conditional capital requirement. In order to formalize this aspect we first recall the notion of the \textit{acceptance set} of a \crm $\rt$: 
\[
\A_t:=\lk X\in\LT\mk \rt(X)\le0\rk.
\]
The following properties of the acceptance set were given in \cite[Proposition 3]{dt5}. 
\begin{proposition}\label{acceptset}
The acceptance set $\A_t$ of a conditional convex risk measure $\rt$ is 
\begin{enumerate}
\item conditionally convex, i.e.\ $\alpha X+(1-\alpha)Y\in\A_t$ for all $X,Y\in\A_t$ and $\alpha$ $\F_t$-measurable such that $0\leq\alpha\leq 1$;
\item solid, i.e. $Y\in\A_t$ whenever $Y\geq X$ for some $X\in\A_t$;
\item such that $0\in\A_t$ and $\ei\lk X\in\Lt\mk X\in\A_t\rk=0$. 
\end{enumerate}
Moreover, $\rt$ is uniquely determined through its acceptance set, since
\begin{equation}\label{defviaset}
\rt(X)=\ei\lk Y\in\Lt\mk X+Y\in\A_t\rk.
\end{equation}
Conversely, if some set $\A_t\subseteq\LT$ satisfies conditions 1)-3), then the functional $\rt\,:\;\LT\rightarrow\Lt$ defined via (\ref{defviaset}) is a conditional convex risk measure.
\end{proposition}
\begin{proof} Properties 1)-3) of the acceptance set follow easily from properties (i)-(iii) in Definition~\ref{defrm}. To prove (\ref{defviaset}) note that by cash invariance $\rt(X)+X\in\A_t$ for all $X$, and this implies ``$\ge$'' in  (\ref{defviaset}). On the other hand, for all $Z\in \lk Y\in\Lt\mk X+Y\in\A_t\rk$ we have 
\[
0\ge\rt(Z+X)=\rt(X)-Z,
\]
thus $\rt(X)\le\ei\lk Y\in\Lt\mk X+Y\in\A_t\rk.$\\
For the proof of the last part of the assertion we refer to \cite[Proposition 3]{dt5}.\end{proof}

Due to (\ref{defviaset}), the value $\rt(X)$ can be viewed as the minimal conditional capital requirement needed to be added to the position $X$ in order to make it acceptable at time $t$.  The following example shows how risk measures can be defined via (\ref{defviaset}).
\begin{example}\label{ex:entr}
Consider the set of all positions having non-negative conditional expected utility, i.e. 
\[
\A_t:=\{X\in\LT\mk E[u_t(X)|\F_t]\geq 0\},
\]
where $u_t$ denotes some non-increasing and concave utility function. It is easy to check that the set $\A_t$ has all 
properties 1)-3) from Proposition~\ref{acceptset}. A basic choice is the exponential utility function
$u_t(x)=1-e^{-\g_tx}$, where $\gamma_t>0$ $P$-a.s.\ denotes the risk aversion parameter such that $\gamma_t,\frac{1}{\gamma_t}\in\Lt$. The corresponding conditional convex risk measure $\rt$ associated to $\A_t$ via \eqref{defviaset} takes the form
\[
\rt(X)=\frac{1}{\g_t}\log E[e^{-\g_t X}|\F_t],\qquad X\in\LT,
\]
and is called the \emph{conditional entropic risk measure}. The entropic risk measure was introduced in \cite{fs4} in the static setting, in the dynamical setting it appeared in \cite{bek4,ms5,dt5,cdk6,fp6,ck6}.  
We characterize the dynamic entropic risk measure in Section~\ref{entropic}.
\end{example}

\section{Robust representation}\label{sectionrr}
As observed in \cite{adeh99, fs4, fr2} in the static setting, the axiomatic properties of a convex risk measure yield, under some regularity condition, a representation of the minimal capital requirement as a suitably modified worst expected loss over a whole class of probabilistic models. In the dynamical setting, such robust representations of conditional coherent risk measures were obtained on a finite probability space in \cite{rse5} for random variables and in \cite{rie4} for stochastic processes.
On a general probability space, robust representations for conditional coherent and convex risk measures
were proved in \cite{dt5,bn4,burg,ks5,fp6,bn8} for random variables and in \cite{cdk6} for stochastic processes.
In this section we mainly summarize the results from \cite{dt5,fp6,bn8}.

The alternative probability measures in a robust representation of a risk measure $\rt$ contribute to the risk evaluation to a different degree. To formalize this aspect we use the notion of the minimal penalty function $\alpha_t^{\min}$, defined for each $Q\in\ma$ as
\begin{equation}\label{pf1}
\pfm=\qes_{X\in\A_t}\ew.
\end{equation}

The following property of the minimal penalty function is a standard result, that will be used in the proof of Theorem~\ref{robdar}.

\begin{lemma}\label{erwpf}
For $Q\in\ma$ and $0\le s\le t$,
\[
E_Q[\pf|\F_s]=\qes_{Y\in\A_t}E_Q[-Y|\F_s]\quad Q\text{-a.s.}
\]
and in particular
\[
E_Q[\pf]=\sup_{Y\in\A_t}E_Q[-Y].
\]
\end{lemma}
\begin{proof} First we claim that the set
\[
\lk E_Q[-X|\F_t]\mk X\in\A_t\rk
\]
is directed upward for any $Q\in\ma$. Indeed, for $X,Y\in\A_t$ we can define $Z:=XI_A+YI_{A^c}$,
where $A:=\{E_Q[-X|\F_t]\ge E_Q[-Y|\F_t]\}\in\F_t$. Conditional convexity of $\rho_t$ implies that 
$Z\in\A_t$, and by definition of $Z$
\[
E_Q[-Z|\F_t]=\max\left(E_Q[-X|\F_t],E_Q[-Y|\F_t]\right)\quad Q\text{-a.s.}.
\]
Hence there exists a sequence $(X^Q_n)_{n\in\nat}$ in $\A_t$ such that 
\begin{equation}\label{folge}
\pf=\lim_nE_Q[-X^Q_n|\F_t]\qquad Q\text{-a.s.},
\end{equation}
and by monotone convergence we get
\begin{align*}
E_Q[\pf|\F_s]&=\lim_nE_Q\left[\,E_Q[-X_n^Q|\F_t]\,\big|\,\F_s\,\right]\\
&\le\qes_{Y\in\A_t}E_Q[-Y|\F_s]\quad Q\text{-a.s.}.
\end{align*}
The converse inequality follows directly from the definition of $\pf$.\end{proof}

The following theorem relates robust representations to some continuity properties of conditional convex risk measures. It combines \cite[Theorem 1]{dt5} with \cite[Corollary 2.4]{fp6}; similar results can be found in \cite{bn4, ks5, cdk6}.

\begin{theorem}\label{robdar}
For a conditional convex risk measure $\rt$ the following are equivalent:
\begin{enumerate}
\item 
$\rt$  has a robust representation
\begin{equation}\label{rd0}
\rt(X)=\esu(\ew-\alpha_t(Q)),\qquad X\in\LT,
\end{equation}
where 
\begin{equation*}
\Q_t:=\lk Q\in\ma\mk Q=P|_{\F_t}\rk
\end{equation*}
and $\alpha_t$ is a map from $\Q_t$ to the set of $\F_t$-measurable random variables with values in $\real\cup\{+\infty\}$, such that $\esu(-\alpha_t(Q))=0$. 
\item
$\rt$  has the robust representation in terms of the minimal penalty function, i.e.
\begin{equation}\label{rd1}
\rt(X)=\esu(\ew-\pf),\qquad X\in\LT,
\end{equation}
where $\alpha_t^{\min}$ is given in (\ref{pf1}).
\item $\rt$  has the robust representation
\begin{equation}\label{rd2}
\rt(X)=\es_{\Q\in\Q^f_t}(\ew-\pf)\quad P\text{-a.s.},\qquad X\in\LT,
\end{equation}
where
\[
\Q_t^f:=\lk Q\in\ma\mk Q=P|_{\F_t}\;E_{Q}[\pf]<\infty\rk.
\]
\item $\rt$ has the ``Fatou-property'': for any bounded sequence $(X_n)_{n\in\nat}$ which converges $P$-a.s.\ to some $X$,
\[
\rt(X)\le\liminf_{n\to\infty}\rt(X_n)\quad\f.
\]
\item
$\rt$ is continuous from above, i.e. 
\[
X_n\searrow X\;\,P\text{-a.s}\quad\Longrightarrow\quad \rt(X_n)\nearrow\rt(X)\;\,P\text{-a.s}
\]
for any sequence $(X_n)_n\subseteq\LT$ and $X\in\LT$.
\end{enumerate}
\end{theorem}
\begin{proof}
3) $\;\ra\; $ 1) and 2) $\;\ra\; $ 1) are obvious.
1) $\,\ra\, $ 4): Dominated convergence implies that $E_Q[X_n|\F_t]\rightarrow E_Q[X|\F_t]$ for each 
$Q\in{\mathcal Q}_t$, and $\liminf_{n\to\infty}\rt(X_n)\ge\rt(X)$ follows by using the robust representation of 
$\rt$ as in the unconditional setting, see, e.g., \cite[Lemma 4.20]{fs4}.

4) $\,\ra\, $ 5): Monotonicity implies $\limsup_{n\to\infty}\rt(X_n)\le\rt(X)$, and $\liminf_{n\to\infty}\rt(X_n)\ge\rt(X)$ follows
by 4).

5) $\,\ra\, $ 2): The inequality
\begin{equation}\label{ungl1}
\rt(X)\ge\es_{Q\in{\mathcal Q}_t}(\ew-\pf)
\end{equation}
follows from the definition of $\alpha_t^{\min}$. In order to prove the equality we will show that
\[
E_P[\rho_t(X)]\le E_P\left[\esu(\ew-\pf)\right].
\]
To this end, consider the map $\rho^P\,:\,\LT\,\rightarrow\,\real$ defined by
$\rho^P(X):=E_P[\rt(X)]$. It is easy to check that $\rho^P$ is a convex risk measure which
is continuous from above. Hence \cite[Theorem 4.31]{fs4} implies that 
$\rho^P$ has the robust representation
\[
\rho^P(X)=\sup_{Q\in\M_1(P)}(E_Q[-X]-\alpha(Q))\qquad X\in\LT,
\]
where the penalty function $\alpha(Q)$ is given by
\[
\alpha(Q)=\sup_{X\in\LT: \rho^P(X)\le0}E_Q[-X].
\]
Next we will prove that $Q\in\Q_t$ if $\alpha(Q)<\infty$. Indeed, let $A\in\F_t$ and $\lambda>0$. Then
\[
-\lambda P[A]=E_P[\rt(\lambda I_A)]=\rho^P(\lambda I_A)\ge E_Q[-\lambda I_A]-\alpha(Q),
\]
so
\[
P[A]\le Q[A]+\frac{1}{\lambda}\alpha(Q)\quad\mbox{for all}\quad \lambda>0,
\]
and hence $P[A]\le Q[A]$ if $\alpha(Q)<\infty$. The same reasoning with $\lambda<0$ implies 
$P[A]\ge Q[A]$, thus $P = Q$ on $\F_t$ if $\alpha(Q)<\infty$. By Lemma~\ref{erwpf}, we have for every $Q\in{\mathcal Q}_t$ 
\[
E_P[\pf]=\sup_{Y\in\A_t}E_P[-Y].
\]
Since $\rho^P(Y)\le0$ for all $Y\in\A_t$, this implies
\begin{equation*}
E_P[\pf]\le\alpha(Q)
\end{equation*}
for all $Q\in{\mathcal Q}_t$, by definition of the penalty function $\alpha(Q)$.

Finally we obtain
\begin{align}\label{rdbeweis}
E_P[\rt(X)]=\rho^P(X)&=\sup_{Q\in\M_1(P), \alpha(Q)<\infty}\left(E_Q[-X]-\alpha(Q)\right)\nonumber\\
&\le\sup_{Q\in\Q_t, E_P[\pf]<\infty}\left(E_Q[-X]-\alpha(Q)\right)\nonumber\\
&\le\sup_{Q\in\Q_t, E_P[\pf]<\infty}E_P[E_Q[-X|\F_t]-\pf]\nonumber\\
&\le E_P\left[\es_{Q\in\Q_t, E_P[\pf]<\infty}\left(E_Q[-X|\F_t]-\pf\right)\right]\\
&\le E_P\left[\es_{Q\in\Q_t}E_Q[-X|\F_t]-\pf\right]\nonumber, 
\end{align}
proving equality (\ref{rd1}).

5) $\,\ra\, $ 3) The inequality
\[
\rt(X)\ge\es_{\Q\in\Q^f_t}(\ew-\pf)
\]
follows from (\ref{ungl1}) since $\Q_t^f\subseteq{\mathcal Q}_t$, and (\ref{rdbeweis}) proves the
equality. 
\end{proof}

The penalty function $\pf$ is minimal in the sense that any other function $\alpha_t$ in a
robust representation \eqref{rd0} of $\rt$ satisfies
\[
\pf\le\alpha_t(Q)\;\,\f
\]
for all $Q\in{\mathcal Q}_t$.
An alternative formula for the minimal penalty function is given by
\begin{equation*}
\pf=\es_{X\in\LT}\,\left(\ew-\rt(X)\right)\quad\mbox{for all}\;\, Q\in{\mathcal Q}_t.
\end{equation*}
This follows as in the unconditional case; see, e.g., \cite[Theorem 4.15, Remark 4.16]{fs4}. 

\begin{remark}\label{abg}
Another characterization of a conditional convex risk measure $\rt$ that is equivalent to the properties 1)-4) of Theorem \ref{robdar} is the following: The acceptance set $\A_t$ is weak$^\ast$-closed, i.e., it is closed in $\LT$ with respect to the topology $\sigma(\LT, L^1(\Omega,\F,P))$.
This equivalence was shown in
\cite{cdk6} in the context of risk measures for processes and in
\cite{ks5} for risk measures for random variables. Though in \cite{ks5} a slightly different definition of a conditional risk measure is used, the reasoning given there works just the same in our case; cf.\ \cite[Theorem 3.16]{ks5}.
\end{remark}

For the characterization of time consistency in Section~\ref{sec:tc} we will need a robust representation of a conditional convex risk measure $\rt$ under any measure $Q\in\ma$, where possibly $Q\notin\Q_t$. Such representation can be obtained as in Theorem~\ref{robdar} by considering $\rt$ as a risk measure under $Q$, as shown in the next corollary. This result is a version of \cite[Proposition 1]{bn8}.

\begin{corollary}\label{corrobdar}
A conditional convex risk measure $\rt$ is continuous from above if and only if it has the robust representations
\begin{align}
\rt(X)&=\qes_{R\in\Q_t(Q)}(E_R\ewo-\pfo(R))\label{rd3}\\
&=\qes_{R\in\Q^f_t(Q)}(E_R\ewo-\pfo(R))\quad Q\text{-a.s.},\quad \forall X\in\LT,\label{rd3f}
\end{align}
for all $Q\in\ma$, where 
\begin{equation*}
\Q_t(Q)=\lk R\in\ma\mk R=Q|_{\F_t}\rk
\end{equation*}
and 
\[
\Q_t^f(Q)=\lk R\in\M_1(P)\mk R=Q|_{\F_t},\;E_{R}[\pfo(R)]<\infty\rk.
\]
\end{corollary}
\begin{proof} To show that continuity from above implies representation \eqref{rd3}, we can replace $P$ by a probability measure $Q\in\ma$ and repeat all the reasoning of the proof of 5)$\ra$2) in Theorem~\ref{robdar}. In this case we  consider the static convex risk measure 
\[
\rho^Q(X)=E_Q[\rt(X)]=\sup_{R\in\M_1(P)}(E_R[-X]-\alpha(R)),\qquad X\in\LT,
\]
instead of $\rho^P$. The proof of \eqref{rd3f} follows in the same way from \cite[Corollary 2.4]{fp6}. Conversely,  continuity from above follows from Theorem~\ref{robdar} since representation \eqref{rd3} holds under $P$.
\end{proof}

\begin{remark}
 One can easily see that the set $\Q_t$ in representations \eqref{rd0} and \eqref{rd1} can be replaced by 
$\Pc_t:=\lk Q\in\ma\mk Q\approx P\;\text{on}\;\F_t\rk$.
 Moreover, representation \eqref{rd0} is also equivalent to
\[
\rt(X)=\es_{Q\in\ma}(\ew-\hat{\alpha}_t(Q)),\qquad X\in\LT,
\]
where the conditional expectation under $Q\in\ma$ is defined under $P$ as
\[
E_Q[X|\F_t]:=\frac{E_P[Z_TX|\F_t]}{Z_t}I_{\{Z_t>0\}},
\]
and the extended penalty function $\hat{\alpha}_t$ is given by 
\begin{eqnarray*}
\hat{\alpha}_t(Q) = \left\{ \begin{array}{ll} \alpha_t(Q) & \textrm{on}\;\{\frac{dQ}{dP}|_{\F_t}>0\}; \\
+\infty & \textrm{otherwise}.
\end{array} \right.
\end{eqnarray*}
\end{remark}

In the \textit{coherent} case the penalty function $\pf$ can only take values $0$ or $\infty$ due to positive homogeneity of $\rt$. Thus representation \eqref{rd3} takes the following form.
\begin{corollary}\label{rdcoherent}
A conditional coherent risk measure $\rt$ is continuous from above if and only if it is 
representable in the form
\begin{equation}\label{rdcoh}
\rt(X)=\es_{\Q\in\Q^0_t}\ew,\qquad X\in\LT,
\end{equation}
where
\[
\Q_t^0:=\lk Q\in\Q_t\mk\pf=0\; Q\mbox{-a.s.}\rk.
\]
\end{corollary}

\begin{example}\label{avar}
A notable example of a conditional coherent risk measure is \emph{conditional Average Value at Risk} defined as
\begin{eqnarray*}
AV@R_{t,\lambda_t}(X):=\es\{E_Q[-X|\F_t]\mk Q\in\Q_t, \frac{dQ}{dP}\leq \lambda_t^{-1}\}
\end{eqnarray*}
with $\lambda_t\in\Lt$, $0<\lambda_t\leq 1$. Static Average Value at Risk was introduced in \cite{adeh99} as a valid alternative to the widely used yet criticized Value at Risk. 
The conditional version of  Average Value at Risk appeared in \cite{adehk7}, and was also studied in \cite{Samuel, vo6}.
\end{example}

As observed, e.g., in \cite[Remark 3.13]{cdk6}, the minimal penalty function has the local property. In our context it means  that for any $Q^1, Q^2\in\Q_t(Q)$ with the corresponding density processes $Z^1$ and $Z^2$ with respect to $P$, and for any $A\in\F_t$, the probability measure $R$ defined via $\frac{dR}{dP}:=I_AZ^1_T+I_{A^{\text{c}}}Z^2_T$ has the penalty function value
\[
\pfo(R)= I_A\pfo(Q^1)+I_{A^{\text{c}}}\pfo(Q^2)\qquad Q\text{-a.s.}.
\]
In particular $R\in\Q_t^f(Q)$ if $Q^1, Q^2\in\Q_t^f(Q)$. Standard arguments (cf., e.g., \cite[Lemma 1]{dt5}) imply then that the set
\[
\lk E_R[\,-X\,|\,\F_t]-\alpha_t^{\min}(R)\mk R\in\Q^f_t(Q)\rk
\]
is directed upward, thus
\begin{equation}\label{erwrho}
E_{Q}[\rt(X)|\F_s]=\qes_{R\in\Q^f_t(Q)}\left(E_{R}[-X|\F_s]-E_{R}[\alpha_t^{\min}(R)|\F_s]\right)
\end{equation}
for all $Q\in\ma, X\in\LTq$ and $0\le s\leq t$. 

\section{Time consistency properties}\label{sec:tc}

In the dynamical setting risk assessment of a financial position is updated when new information is released. This leads to the notion of a dynamic risk measure. 

\begin{definition}\label{dcrm}
A a sequence $(\rt)\zt$ is called a \emph{dynamic convex risk measure} if $\rt$ is a conditional convex risk measure for each $t\in\T$.
\end{definition}

A key question in the dynamical setting is how the conditional risk assessments at different times are interrelated. This question has led to several notions of time consistency discussed in the literature. 
A unifying view was suggested in \cite{tu8}. 
\begin{definition}\label{sina}
Assume that $(\rt)\zt$ is a dynamic convex risk measure and let $\Y_t$ be a subset of $\LT$ such that $0\in\Y_t$ and $\Y_t+\real=\Y_t$ for each $t\in\T$. Then $(\rt)\zt$ is called \emph{acceptance (resp. rejection) consistent with respect to $(\Y_t)\zt$}, if for all $t\in\T$ such that $t<T$ and for any $X\in\LT$ and $Y\in\Y_{t+1}$ the following condition holds:
\begin{equation}\label{definition1}
\rho_{t+1}(X)\le\rho_{t+1}(Y)\;\;(\mbox{resp.}\,\ge)\quad\Longrightarrow\quad\rho_{t}(X)\le\rho_{t}(Y)\;\;(\mbox{resp.}\,\ge).
\end{equation}
\end{definition}
The idea is that the degree of time consistency is determined by a sequence of benchmark sets $(\Y_t)\zt$: if a financial position at some future time is always preferable to some element of the benchmark set, then it should also be preferable today. The bigger the benchmark set, the stronger is the resulting notion of time consistency. In the following we focus on three cases.

\begin{definition}\label{cons}
We call a dynamic convex risk measure $(\rt)\zt$
\begin{enumerate}
\item \emph{strongly time consistent}, if it is either acceptance consistent or rejection consistent with respect to $\Y_t=\LT$ for all $t$ in the sense of Definition~\ref{sina};
\item \emph{middle acceptance (resp. middle rejection) consistent}, if for all $t$ we have $\Y_t=L^\infty_t$ in Definition~\ref{sina};
\item \emph{weakly acceptance (resp. weakly rejection) consistent}, if for all $t$ we have $\Y_t=\real$ in Definition~\ref{sina}.
\end{enumerate}
\end{definition}
Note that there is no difference between rejection consistency and acceptance consistency with respect to $\LT$, since the role of $X$ and $Y$ is symmetric in that case. Obviously strong time consistency implies both middle rejection and middle acceptance consistency, and middle rejection (resp. middle acceptance) consistency implies weak rejection (resp. weak acceptance) consistency. In the rest of the paper we drop the terms ``middle'' and ``strong'' in order to simplify the terminology.

\subsection{Time consistency}\label{subsec:tc}
Time consistency has been studied extensively in the recent work on dynamic risk measures, see \cite{adehk7,d6,rie4,dt5,cdk6,ks5,burg,bn8,ipen7,fp6,ck6, dpr10} and the references therein. In the next proposition we recall some equivalent characterizations of time consistency. 

\begin{proposition}\label{def2}
A dynamic convex risk measure $(\rt)\zt$ is time consistent if and only if any of the following conditions holds:
\begin{enumerate}
\item for all $t\in\T$ such that $t<T$ and for all $X,Y\in\LT$:
\begin{equation}\label{tc4}
\rtt(X)\le\rtt(Y)\;\,P\text{-a.s}\quad\Longrightarrow\quad\rt(X)\le\rt(Y)\;\,P\text{-a.s.};
\end{equation}
\item for all $t\in\T$ such that $t<T$ and for all $X,Y\in\LT$:
\begin{equation}\label{tc2}
\rtt(X)=\rtt(Y)\;\,P\text{-a.s}\quad\Longrightarrow\quad\rt(X)=\rt(Y)\;\,P\text{-a.s.};
\end{equation}
\item
$(\rt)\zt$ is recursive, i.e.
\[
\rt=\rt(-\rs)\quad P\text{-a.s.}
\]
for all $t,s\ge 0$ such that $t,t+s\in\T$.
\end{enumerate}
\end{proposition}
\begin{proof}
It is obvious that time consistency implies condition (\ref{tc4}), and that (\ref{tc4}) implies (\ref{tc2}). By cash invariance we have $\rtt(-\rtt(X))=\rtt(X)$ and hence one-step recursiveness follows from (\ref{tc2}). We prove that one-step recursiveness implies recursiveness by induction on $s$. For $s=1$ the claim is true for all $t$. Assume that the induction hypothesis holds for each $t$ and all $k\le s$ for some $s\ge 1$. Then we obtain
\begin{align*}
\rt(-\rho_{t+s+1}(X))&=\rt(-\rs(-\rho_{t+s+1}(X)))\\
&=\rt(-\rs(X))\\
&=\rt(X),
\end{align*}
where we have applied the induction hypothesis to the random variable $-\rho_{t+s+1}(X)$. Hence the claim
follows. Finally, due to monotonicity, recursiveness implies time consistency.
\end{proof}

If we restrict a conditional convex risk measure $\rt$ to the space $L^\infty_{t+s}$ for some $s\ge0$, the corresponding acceptance set is given by
\[
\A_{t,t+s}:=\lk X\in L^\infty_{t+s}\mk\rt(X)\le0\;\,P\text{-a.s.}\rk,
\]
and the minimal penalty function by
\begin{equation}\label{ats}
\pfts:=\qes_{X\in\A_{t,t+s}}\,\ew, \qquad Q\in\ma.
\end{equation}

The following lemma recalls equivalent characterizations of recursive inequalities in terms of acceptance sets from \cite[Lemma 4.6]{fp6};  property \eqref{eqacset1} was shown in \cite{d6}.

\begin{lemma}\label{4.6}
Let $(\rt)\zt$ be a dynamic convex risk measure. Then the following equivalences hold for all 
$s,t$ such that $t,t+s\in\T$ and all $X\in\LT$:
\begin{align}
X\in\A_{t,t+s}+\A_{t+s}&\iff-\rho_{t+s}(X)\in\A_{t,t+s}\label{eqacset1}\\
\A_t\subseteq \A_{t,t+s}+\A_{t+s}&\iff\rt(-\rho_{t+s})\le\rt\quad\f\label{eqacset2}\\
\A_t\supseteq \A_{t,t+s}+\A_{t+s}&\iff\rt(-\rho_{t+s})\ge\rt\quad\f.\label{eqacset3}
\end{align}
\end{lemma}

\begin{proof}
To prove ``$\ra$'' in (\ref{eqacset1}) let $X=X_{t,t+s}+X_{t+s}$ with $X_{t,t+s}\in\A_{t,t+s}$ and
$X_{t+s}\in\A_{t+s}$. Then
\[
\rho_{t+s}(X)=\rho_{t+s}(X_{t+s})-X_{t,t+s}\le-X_{t,t+s}
\]
by cash invariance, and monotonicity implies
\[
\rt(-\rho_{t+s}(X))\le\rt(X_{t,t+s})\le0.
\]
The converse direction follows immediately from
$X=X+\rho_{t+s}(X)-\rho_{t+s}(X)$ and $X+\rho_{t+s}(X)\in\A_{t+s}$ for all $X\in\LT$.

In order to show ``$\ra$'' in (\ref{eqacset2}), fix $X\in\LT$. Since 
$X+\rt(X)\in\A_t\subseteq \A_{t,t+s}+\A_{t+s}$, we obtain
\[
\rho_{t+s}(X)-\rt(X)=\rho_{t+s}(X+\rt(X))\in-\A_{t,t+s},
\]
by (\ref{eqacset1}) and cash invariance. Hence
\[
\rt(-\rho_{t+s}(X))-\rt(X)=\rt(-(\rho_{t+s}(X)-\rt(X)))\le0\quad\f.
\]
To prove ``$\Leftarrow$'' let $X\in\A_t$. Then $-\rho_{t+s}(X)\in\A_{t,t+s}$ by the right hand side
of (\ref{eqacset2}), and hence $X\in\A_{t,t+s}+\A_{t+s}$ by (\ref{eqacset1}).

Now let $X\in\LT$ and assume $\A_t\supseteq \A_{t,t+s}+\A_{t+s}$. Then
\begin{align*}
\rt(-\rho_{t+s}(X))+X&=\rt(-\rho_{t+s}(X))-\rho_{t+s}(X)+\rho_{t+s}(X)+X\\&\in\A_{t,t+s}+\A_{t+s}\subseteq\A_t.
\end{align*}
Hence 
\[
\rt(X)-\rt(-\rho_{t+s}(X))=\rt(X+\rt(-\rho_{t+s}(X)))\le 0
\]
by cash invariance, and this proves ``$\ra$'' in (\ref{eqacset3}). For the converse direction let 
$X\in\A_{t,t+s}+\A_{t+s}$. Since $-\rho_{t+s}(X)\in\A_{t,t+s}$ by (\ref{eqacset1}), we obtain
\[
\rt(X)\le\rt(-\rho_{t+s}(X))\le0,
\]
hence $X\in\A_t$.\end{proof}

We also have the following relation between acceptance sets and penalty functions; cf.\ \cite[Lemma 2.2.5]{ipen7}.
\begin{lemma}\label{setpen}
Let $(\rt)\zt$ be a dynamic convex risk measures. Then the following implications hold for all 
$t,s$ such that $t,t+s\in\T$ and for all $Q\in\ma$:
\begin{align*}
\A_t\subseteq \A_{t,t+s}+\A_{t+s}&\Rightarrow\pf\le\pfts+E_Q[\pfs|\F_t]\quad\qf
\\
\A_t\supseteq \A_{t,t+s}+\A_{t+s}&\Rightarrow\pf\ge\pfts+E_Q[\pfs|\F_t]\quad\qf.
\end{align*}
\end{lemma}
\begin{proof} Straightforward from the definition of the minimal penalty function and Lemma~\ref{erwpf}. 
\end{proof}

The following theorem gives equivalent characterizations of time consistency in terms of acceptance sets, penalty functions, and a supermartingale property of the risk process.

\begin{theorem}\label{eqchar}
Let $(\rt)\zt$ be a dynamic convex risk measure such that each $\rt$ is \ctse. Then the
following conditions are equivalent:
\begin{enumerate}
\item $(\rt)\zt$ is time consistent.
\item $\A_t=\A_{t,t+s}+\A_{t+s}\,$ for all $t,s$ such that $t,t+s\in\T$.
\item $\pf=\pfts+E_Q[\,\pfs\,|\,\F_t\,]\;\; Q$-a.s. for all $t,s$ such that $t,t+s\in\T$ and all $\,Q\in\ma$.
\item For all $X\in\LTq$ and  all $t,s$ such that $t,t+s\in\T$ and all $\,Q\in\ma$ we have
\[
 E_Q[\,\rho_{t+s}(X)+\pfs\,|\,\F_t]\le\rt(X)+\pf\quad\qf.
\]

\end{enumerate}
\end{theorem} 

Equivalence of properties 1) and 2) of Theorem~\ref{eqchar} was proved in \cite{d6}. Characterizations of time consistency in terms of penalty functions as in 3) of Theorem~\ref{eqchar} appeared in \cite{fp6, bn6, ck6, bn8}; similar results for risk measures for processes were given in \cite{cdk6, ck6}. The supermartingale property  as in 4) of Theorem~\ref{eqchar} was obtained in \cite{fp6}; cf.\ also \cite{bn8} for the absolutely continuous case. 

\begin{proof} The proof of 1)$\ra$2)$\ra$3) follows from  Lemma~\ref{4.6} and Lemma~\ref{setpen}. To prove 3)$\ra$4) fix $Q\in\ma$. 
By \eqref{erwrho} we have
\[
E_{Q}[\rs(X)|\F_t]=\qes_{R\in\Q^f_{t+s}(Q)}\left(E_{R}[-X|\F_{t}]-E_{R}[\alpha_{t+s}^{\min}(R)|\F_t]\right). 
\]
On the set $\lk\pf=\infty\rk$ property 4) holds trivially. On the set $\lk\pf<\infty\rk$ property 3) implies 
$E_Q[\alpha_{t+s}^{\min}(Q)|\F_t]<\infty$ and $\alpha_{t,t+s}^{\min}(Q)<\infty$, then for $R\in\Q^f_{t+s}(Q)$
\begin{equation*}
\pfo(R)=\alpha_{t,t+s}^{\min}(Q)+E_R[\alpha_{t+s}^{\min}(R)|\F_t]<\infty\quad\qf.
\end{equation*}
Thus
\begin{equation*}
E_{Q}[\rs(X)+\alpha_{t+s}^{\min}(Q)|\F_t]=\qes_{R\in\Q^f_{t+s}(Q)}\left(E_{R}[-X|\F_{t}]-\alpha_{t}^{\min}(R)\right)+\pf
\end{equation*}
on $\lk\pf<\infty\rk$.
Moreover, since $\Q_{t+s}^f(Q)\subseteq \Q_{t}(Q)$, \eqref{rd3} implies
\begin{equation*}
E_{Q}[\rs(X)+\alpha_{t+s}^{\min}(Q)|\F_t]\le\qes_{R\in \Q_{t}(Q)}\left(E_{R}[-X|\F_{t}]-\alpha_{t}^{\min}(R)\right)+\pf=\rt(X)+\pf\quad\qf.
\end{equation*}
It remains to prove 4)$\ra$1). To this end fix $Q\in\Q_t^f$ and $X,Y\in\LT$ such that $\rho_{t+1}(X)\le \rho_{t+1}(Y)$. Note that $E_Q[\alpha_{t+s}(Q)]<\infty$ due to 4), hence $Q\in\Q_{t+s}^f(Q)$.
Using 4) and representation \eqref{rd3f} for $\rho_{t+s}$ under $Q$, we obtain
\begin{align*}
\rt(Y)+\alpha_t^{\min}(Q)&\geq E_Q[\rho_{t+1}(Y)+\alpha_{t+1}^{\min}(Q)|\F_t]\\
&\geq E_Q[\rho_{t+1}(X)+\alpha_{t+1}^{\min}(Q)|\F_t]\\
&\geq E_Q[E_Q[-X|\F_{t+1}]-\alpha_{t+1}^{\min}(Q)+\alpha_{t+1}^{\min}(Q)|\F_t]\\
&=E_Q[-X|\F_t].
\end{align*}
Hence  representation \eqref{rd2} yields $\rho_t(y)\ge \rho_t(X)$, and time consistency follows from Proposition~\ref{def2}.
\end{proof}
Properties 3) and 4) of Theorem~\ref{eqchar} imply in particular supermartingale propeties of penalty function processes and risk processes. This allows to apply martingale theory for characterization the the dynamics of these processes, as we do in Proposition~\ref{worstcase} and Proposition~\ref{riesz}; cf.\ also \cite{d6, fp6, ipen7, bn8, dpr10}. 
\begin{proposition}\label{worstcase}
Let $(\rt)\zt$ be a time consistent dynamic convex risk measure such that each $\rt$ is \ctse. Then the process
\begin{equation*}
V_t^Q(X):=\rt(X)+\pf,\qquad t\in\T
\end{equation*}
is a $Q$-supermartingale for all $X\in\LT$ and all $Q\in\Q_0$, where
\begin{equation*}
\Q_0:=\lk Q\in\ma\mk \pfn<\infty\rk.
\end{equation*} 
Moreover, $(V_t^Q(X))_{t\in\T}$ is a $Q$-martingale if $Q\in\Q_0$ is  a ``worst case'' measure for $X$ at time $0$, i.e.\ if the supremum in the robust representation of $\rho_0(X)$ is attained at $Q$:
\[
\rho_0(X)=E_Q[-X]-\alpha^{\min}_0(Q)\quad Q\text{-a.s.}.
\] 
In this case $Q$ is a ``worst case'' measure for $X$ at any time $t$, i.e. 
\[
\rho_t(X)=E_Q[-X|\F_t]-\alpha^{\min}_t(Q)\quad Q\text{-a.s.}\quad\text{for all}\quad t\in\T.
\]
The converse holds if $T<\infty$ or $\lim_{t\to\infty}\rt(X)=-X$ $P$-a.s.\ (what is called asymptotic precision in \cite{fp6}): If $(V_t^Q(X))_{t\in\T}$ is a $Q$-martingale then  $Q\in\Q_0$ is a ``worst case'' measure for $X$ at any time $t\in\T$.
\end{proposition}
\begin{proof}
The supermartingale property of $(V_t^Q(X))\zt$ under each $Q\in\Q_0$ follows directly from properties 3) and 4) of Theorem~\ref{eqchar}. To prove the remaining part of the claim, fix $Q\in\Q_0$ and $X\in\LT$.
If $Q$ is a ``worst case'' measure for $X$ at time $0$, the process
\[
U_t(X):=V_t^Q(X)-E_Q[-X|\F_t],\qquad t\in\T
\]
is a non-negative $Q$-supermartingale beginning at $0$. Indeed, the supermartingale property follows from that of $(V_t^Q(X))\zt$, and  non-negativity follows from the representation \eqref{rd3f}, since $Q\in\Q_t^f(Q)$. Thus $U_t=0$ $Q$-a.s.\ for all $t$, and this proves the ``if'' part of the claim. To prove the converse direction, note that if $(V_t^Q(X))_{t\in\T}$ is a $Q$-martingale and $\rho_T(X)=-X$ (resp.\ $\lim_{t\to\infty}\rt(X)=-X$ $P$-a.s.), the process $U(X)$ is a $Q$-martingale ending at $0$ (resp.\ converging to $0$ in $L^1(Q)$), and thus $U_t(X)=0$ $Q$-a.s.\ for all $t\in\T$.
\end{proof}

\begin{remark}
The fact that a worst case measure for $X$ at time $0$, if it exists, remains  a worst case measure for $X$ at any time $t\in\T$ was also shown in \cite[Theorem 3.9]{ck6} for a time consistent dynamic risk measure without using the supermartingale property from Proposition~\ref{worstcase}.
\end{remark}

\begin{remark}\label{remarksm}
In difference to \cite[Theorem 4.5]{fp6}, without the additional assumption that the set
\begin{equation}\label{qstar}
\Qs:=\lk Q\in\me\mk \pfn<\infty\rk
\end{equation}
is nonempty, the supermartingale property of $(V_t^Q(X))\zt$ for all $X\in\LT$ and all $Q\in\Qs$ is not sufficient to prove time consistency.  In this case we also do not have the robust representation of $\rt$ in terms of the set $\Qs$.
\end{remark}

The process $(\pf)\zt$ is a $Q$-supermartingale for all $Q\in\Q_0$ due to Property 3) of Theorem~\ref{eqchar}. The next proposition provides the explicit form of its Doob- and its Riesz-decomposition; cf.\ also \cite[Proposition 2.3.2]{ipen7}. 

\begin{proposition}\label{riesz}
Let $(\rt)\zt$ be a time consistent dynamic convex risk measure such that each $\rt$ is \ctse. Then for each $Q\in\Q_0$ the process $(\pf)\zt$ is a non-negative $Q$-supermartingale with the Riesz decomposition
\[
\pf=Z_t^Q+M_t^Q\quad \qf,\qquad t\in\T,
\]
where
\[
Z_t^Q:= E_Q\left[\,\sum_{k=t}^{T-1}\pfk\,\big|\,\F_t\,\right]\quad \qf,\quad t\in\T
\]
is a $Q$-potential and
\[
M_t^{Q}:=\left\{
\begin{array}{c@{\quad  \quad}l}
0 & \text{if $T<\infty$},\\
\displaystyle\lim_{s\to\infty}E_Q\left[\alpha_s(Q)\,|\,\F_t\,\right] & \text{if $T=\infty$}
\end{array}\right.\qquad \qf,\quad t\in\T
\]
is a non-negative $Q$-martingale. 

Moreover, the Doob decomposition of $(\pf)\zt$ is given by 
\begin{equation*}
\pf=E_Q\left[\,\sum_{k=0}^{T-1}\pfk\,\big|\,\F_t\,\right]+M_t^Q-\sum_{k=0}^{t-1}\pfk,\quad t\in\T
\end{equation*}
with the $Q$-martingale
\[
E_Q\left[\,\sum_{k=0}^{T-1}\pfk\,\big|\,\F_t\,\right]+M_t^Q,\quad t\in\T
\]
and the non-decreasing predictable process $(\sum_{k=0}^{t-1}\pfk)\zt$.
\end{proposition}
\begin{proof} We fix $Q\in\ma$ and applying property 3) of Theorem \ref{eqchar} step by step we obtain
\begin{equation}\label{lim}
\pf=E_Q\left[\,\sum_{k=t}^{t+s-1}\pfk\,\big|\,\F_t\,\right]+E_Q[\,\pfs\,|\,\F_t\,]\quad \qf
\end{equation}
for all $t,s$ such that $t,t+s\in\T$. If $T<\infty$, the Doob- and Riesz-decompositions follow immediately from \eqref{lim}, since $\alpha_T(Q)=0\; \qf$. If $T=\infty$, by monotonicity there exists the limit
\[
Z_t^Q= \lim_{s\to\infty}E_Q\left[\,\sum_{k=t}^{s}\pfk\,\big|\,\F_t\,\right]=E_Q\left[\,\sum_{k=t}^{\infty}\pfk\,\big|\,\F_t\,\right]\quad \qf
\]
for all $t\in\T$, where we have used the monotone convergence theorem for the second equality. Equality \eqref{lim} implies then that there exists 
\[
M_t^Q= \lim_{s\to\infty}E_Q\left[\,\pfs\,|\,\F_t\,\right]\quad \qf,\quad t\in\T
\]
and
\[
\pf=Z_t^Q+M_t^Q\quad \qf
\]
for all $t\in\T$.

The process $(Z_t^Q)\zt$ is a non-negative $Q$-supermartingale. Indeed, 
\begin{equation}\label{fin}
E_Q[\,Z_{t}^Q\,]\le E_Q\left[\,\sum_{k=0}^{\infty}\pfk\,\right]\le\pfn<\infty
\end{equation}
and $E_Q[\,Z_{t+1}^Q\,|\,\F_t\,]\le Z_t^Q$ $Q$-a.s. for all $t\in\T$ by definition. Moreover, monotone convergence implies
\[
\lim_{t\to\infty}E_Q[\,Z_t^Q\,]=E_Q\left[\,\lim_{t\to\infty}\sum_{k=t}^{\infty}\pfk\,\right]=0\quad\qf,
\]
since $\sum_{k=0}^{\infty}\pfk<\infty\;Q$-a.s. by \eqref{fin}. Hence the process $(Z_t^Q)\zt$ is a $Q$-potential.

The process $(M_t^Q)\zt$ is a non-negative $Q$-martingale, since
\[
E_Q[\,M_{t}^Q\,]\le E_Q\left[\,\pf\,\right]\le\pfn<\infty
\]
and 
\begin{align*}
E_Q[M_{t+1}^Q-M_{t}^Q|\F_t]&=E_Q[\pft|\F_t]-\pf-E_Q[Z_{t+1}^Q-Z_t^Q|\F_t]\\
&=\pftt-\pftt=0\qquad\Q\mbox{-a.s.}
\end{align*}
for all $t\in\T$ by property 3) of Theorem \ref{eqchar} and the definition of $(Z_t^Q)\zt$.

The Doob-decomposition follows straightforward from the Riesz-decomposition.
\end{proof}

\begin{remark}\label{mnull}
It was shown in \cite[Theorem 5.4]{fp6} that the martingale $M^Q$ in the Riesz decomposition of $(\pf)_{t\in\T}$ vanishes if and only if $\lim_{t\to\infty}\rt(X)\ge-X\,P$-a.s., i.e.\ the dynamic risk measure $(\rt)\zt$ is asymptotically safe. This is not always the case; see \cite[Example 5.5]{fp6}. 

\end{remark}

For a \emph{coherent} risk measure we have
\begin{equation*}
\Q_t^f(Q)=\Q_t^0(Q):=\lk R\in\M^1(P)\mk R=Q|_{\F_t},\;\; \pfo(R)=0\;\qf\rk.
\end{equation*}
In order to give an equivalent characterization of property 3) of Theorem~\ref{eqchar} in the coherent case, we introduce the sets
\begin{equation*}
\Q_{t,t+s}^0(Q)=\lk R\ll P|_{\F_{t+s}}\mk R=Q|_{\F_t},\;\; \alpha_{t,t+s}^{\min}(R)=0\;\qf\rk\quad \forall\, t,s\ge 0\;\, \textrm{such that}\;\, t,t+s\in\T.
\end{equation*}
For $Q^1\in\Q_{t, t+s}^0(Q)$ and $Q^2\in\Q_{t+s}^0(Q)$ we denote by $Q^1\oplus^{t+s} Q^2$ the pasting  of $Q^1$ and $Q^2$ in $t+s$ via $\Omega$, i.e. the measure $\tq$ defined via
\begin{equation}\label{pasting}
\tq(A)=E_{Q^1}\left[E_{Q^2}[I_A|\F_{t+s}]\right],\qquad A\in\F.
\end{equation}
The relation between stability under pasting and time consistency of coherent risk measures that can be represented in terms of equivalent probability measures was studied in \cite{adehk7, d6, ks5, fp6}. In our present setting, Theorem~\ref{eqchar} applied to a coherent risk measure takes the following form.
\begin{corollary}\label{coherent}
Let $(\rt)\zt$ be a dynamic coherent risk measure such that each $\rt$ is \ctse. Then the following conditions are equivalent:
\begin{enumerate}
\item $(\rt)\zt$ is time consistent.
\item For all $Q\in\ma$ and all $t,s$ such that $t,t+s\in\T$
\begin{equation*}
\Q_t^0(Q)= \lk Q^1\oplus^{t+s} Q^2\mk Q^1\in\Q_{t, t+s}^0(Q), \;Q^2\in\Q_{t+s}^0(Q^1)\rk.
\end{equation*}
\item For all $Q\in\ma$ such that $\pf=0\;Q$-a.s.,
\[
E_Q[\rs(X)\,|\,\F_t] \le \rt(X)\quad\text{and}\quad \alpha_{t+s}^{\min}(Q)=0\;\,\qf
\]
for all $X\in\LTq$ and for all $t,s$ such that $t,t+s\in\T$. 
\end{enumerate}
\end{corollary}
\begin{proof}
$1)\ra 2)$: Time consistency implies property 3) of Theorem \ref{eqchar}, and we will show that this implies property 2) of Corollary \ref{coherent}. Fix $Q\in\ma$. To prove ``$\supseteq$''  let $Q^1\in\Q_t^0(Q)$,  $Q^2\in\Q_{t+s}^0(Q^1)$, and consider $\tq$ defined as in \eqref{pasting}. Note that $\tq=Q^1$ on $\F_{t+s}$ and  
\[
E_{\tq}[X|\F_{t+s}]=E_{Q^2}[X|\F_{t+s}]\quad Q^1\text{-a.s. for all}\quad X\in\LTq.
\]
Hence, using 3) of Theorem~\ref{eqchar} we obtain
\begin{align*}
\alpha_t^{\min}(\tq)&= \alpha_{t,t+s}^{\min}(\tq)+E_{\tq}[\alpha_{t+s}^{\min}(\tq)|\F_{t}]\\
&= \alpha_{t,t+s}^{\min}(Q^1)+E_{Q^1}[\alpha_{t+s}^{\min}(Q^2)|\F_{t}]=0\qquad\qf,
\end{align*}
and thus $\tq \in\Q_t^0(Q)$.
Conversely, for every $\tq\in\Q_t^0(Q)$ we have $\alpha_{t+s}^{\min}(\tq)=\alpha_{t,t+s}^{\min}(\tq)=0\;\tq$-a.s. by 3) of Theorem~\ref{eqchar}, and $\tq=\tq\oplus\tq$. This proves ``$\subseteq$''.

$2)\ra 3)$: Let $R\in\ma$ with $\pfo(R)=0\;R$-a.s.. Then $R\in\Q_t^0(R)$, and thus $R=Q^1\otimes^{t+s} Q^2$ for some $Q^1\in\Q_{t,t+s}^0(R)$ and $Q^2\in\Q_{t+s}^0(Q^1)$. This implies $R=Q^1$ on $\F_{t+s}$ and 
\[
E_R[X|\F_{t+s}]=E_{Q^2}[X|\F_{t+s}]\quad R\text{-a.s.}. 
\]
Hence $\alpha_{t, t+s}^{\min}(R)=\alpha_{t, t+s}^{\min}(Q^1)=0\;\,R$-a.s., and $\alpha_{t+s}^{\min}(R)=\alpha_{t+s}^{\min}(Q^2)=0\,R$-a.s..
To prove the inequality 3) note that due to \eqref{erwrho}
\begin{align*}
E_R[\,\rs(X)\,|\,\F_t]&=\res_{Q\in\Q_{t+s}^0(R)} E_{Q}[-X\,|\,\F_t]\\
&\le \res_{Q\in\Q_t^0(R)}\ew=\rt(X)\quad R\text{-a.s.},
\end{align*}
where we have used that the pasting of $R|_{\F_{t+s}}$ and $Q$ belongs to $\Q_t^0(R)$.

$3)\ra 1)$: Obviously property 3) of Corollary~\ref{coherent} implies property 4) of Theorem~\ref{eqchar} and thus time consistency. \end{proof}

\subsection{Rejection and acceptance consistency}\label{subsec:rc}
Rejection and acceptance consistency were introduced and studied in \cite{tu8, Samuel, ipen7}. These properties can be characterized via recursive inequalities as stated in the next proposition; see \cite[Theorem 3.1.5]{tu8} and \cite[Proposition 3.5]{Samuel}.

\begin{proposition}\label{rejrecursiveness}
A dynamic convex risk measure $(\rt)\zt$ is rejection (resp. acceptance) consistent if and only if
for all $t\in\T$ such that $t<T$ 
\begin{equation}\label{rcdef}
\rt(-\rho_{t+1})\le\rt\quad(\mbox{resp.}\ge)\quad\f.
\end{equation}
\end{proposition}

\begin{proof} We argue for the case of  rejection consistency; the case of  acceptance consistency follows in the same manner. Assume first that $\rts$ satisfies (\ref{rcdef}) and let $X\in\LT$ and $Y\in L^{\infty}(\F_{t+1})$ such that $\rho_{t+1}(X)\ge\rho_{t+1}(Y)$. Using cash invariance, (\ref{rcdef}), and monotonicity, we obtain 
\[
\rt(X)\ge\rt(-\rho_{t+1}(X))\ge\rt(-\rho_{t+1}(Y))=\rt(Y).
\]
The converse implication follows due to cash invariance by applying (\ref{definition1}) to $Y=-\rho_{t+1}(X)$.
\end{proof}

\begin{remark}\label{weakmiddle}
For a dynamic \emph{coherent} risk measure, weak acceptance consistency and acceptance consistency are equivalent. This was shown in \cite[Proposition 3.9]{Samuel}. 
\end{remark}

Another way to characterize  rejection consistency was suggested in \cite{ipen7}.
\begin{proposition}
A dynamic convex risk measure $(\rt)\zt$ is  rejection consistent if only if any of the following conditions holds:
\begin{enumerate}
 \item For all $t\in\T$ such that $t<T$ and all $X\in\LT$  
 \begin{equation}\label{pruddef}
\rt(X)-\rho_{t+1}(X) \in\A_{t,t+1};
\end{equation}
\item For all $t\in\T$ such that $t<T$ and all $X\in\A_t$, we have $-\rtt(X)\in\A_t$.
\end{enumerate}
\end{proposition}
\begin{proof} Since 
\begin{equation*}
\rt(-\rtt(X))=\rt(\rt(X)-\rtt(X))+\rt(X)
\end{equation*}
by cash invariance, \eqref{pruddef} implies rejection consistency, and obviously rejection consistency implies condition 2). If 2) holds, then for any $X\in\LT$ 
\begin{equation*}
\rt(\rt(X)-\rtt(X))=\rt\left(-\rtt(X+\rt(X))\right)\le0,
\end{equation*}
due to cash invariance and the fact that $X+\rt(X)\in\A_t$. 
\end{proof}

Property \eqref{pruddef} was introduces in \cite{ipen7} under the name \emph{prudence}. It means that the adjustment  
$\rtt(X)-\rt(X)$ of the minimal capital requirement for $X$ at time $t+1$ is acceptable at time $t$. In other words, one stays on the safe side at each period of time by making capital reserves according to a rejection consistent dynamic risk measure. 

Similar to time consistency, rejection and acceptance consistency can be characterized in terms of acceptance sets and penalty functions.  

\begin{theorem}\label{eqcharprud}
Let $(\rt)\zt$ be a dynamic convex risk measure such that each $\rt$ is \ctse. Then the following properties are equivalent:
\begin{enumerate}
\item
$(\rt)\zt$ is  rejection consistent (resp.  acceptance consistent).
\item The inclusion
\[
\A_t\subseteq\A_{t,t+1}+\A_{t+1}\quad\text{resp.}\quad\A_t\supseteq\A_{t,t+1}+\A_{t+1}
\] 
holds for all $t\in\T$ such that $t<T$.
\item  The inequality
\[
\pf\le(\text{resp.}\ge)\pftt+E_Q[\,\pft\,|\,\F_t\,]\quad\qf
\] 
holds for all $t\in\T$ such that $t<T$ and all $Q\in\ma$.
\end{enumerate}
\end{theorem}
\begin{proof}  Equivalence of 1) and 2) was proved in Proposition~\ref{rejrecursiveness} and  Lemma~\ref{4.6}, and the proof of  $2)\ra 3) $ is given in Lemma~\ref{setpen}. 

Let us show that property 3) implies property 1). We argue for the case of  rejection consistency; the case of  acceptance consistency follows in the same manner. We fix $t\in\T$ such that $t<T$, and consider the risk measure
\[
{\widetilde \rho_t(X)}:=\rt(-\rtt(X)),\qquad X\in\LT.
\]
It is easily seen that ${\widetilde \rho_t}$ is a conditional convex risk measure that is \ctse. Moreover, the dynamic risk measure $({\widetilde \rho_t}, \rtt)$ is time consistent by definition, and thus it fulfills properties 2) and 3) of Theorem~\ref{eqchar}. We denote by ${\widetilde \A_t}$ and 
${\widetilde \A_{t,t+1}}$ the acceptance sets of the risk measure ${\widetilde \rho_t}$, and by ${\widetilde \alpha_{t}^{\min}}$ its penalty function. Since 
\[
{\widetilde \rho_t(X)}=\rt(-\rtt(X))=\rt(X)
\]
for all $X\in L_{t+1}$, we have ${\widetilde \A_{t,t+1}}=\Att$,
and thus
\[
{\widetilde \A_t}=\Att+\At
\]
by 2) of Theorem~\ref{eqchar}. Lemma~\ref{setpen} and property 3) then imply
\begin{equation*}
{\widetilde \alpha_{t}^{\min}}(Q) =\pftt+E_Q[\pft|\F_t]\ge\pf
\end{equation*}
for all $Q\in\Q_t$.  Thus  
\begin{equation*}
\rt(X)\ge{\widetilde \rho_t(X)}=\rt(-\rtt(X))
\end{equation*}
for all $X\in\LT$, due to representation (\ref{rd2}). 
\end{proof}

\begin{remark}\label{corm}
Similar to Corollary~\ref{coherent}, condition 3) of Theorem~\ref{eqcharprud} can be restated for a dynamic \emph{coherent} risk measure $(\rt)\zt$ as follows: 
\begin{equation*}
\Q_t^0(Q) \supseteq \lk Q^1\oplus^{t+1} Q^2\mk Q^1\in\Q_{t, t+1}^0(Q), \;Q^2\in\Q_{t+1}^0(Q^1)\rk\quad (\mbox{resp.} \subseteq)
\end{equation*}
for all $t\in\T$ such that $t<T$ and all $Q\in\ma$.
\end{remark}

The following proposition provides an additional equivalent characterization of rejection consistency, that can be viewed as an analogon of the supermartingale property 4) of Theorem~\ref{eqchar}.
\begin{proposition}\label{mrc}
Let $(\rt)\zt$ be a dynamic convex risk measure such that each $\rt$ is \ctse. Then $(\rt)\zt$ is rejection consistent if and only if the inequality
\begin{equation}\label{superm1}
E_Q\left[\,\rho_{t+1}(X)\,|\,\F_t\,\right]\le \rho_t(X)+\pftt\qquad Q\text{-a.s.}
\end{equation}
holds for all  $Q\in\ma$  and all $t\in\T$ such that $t<T$. In this case the process 
\[
U_t^Q(X):=\rt(X)-\sum_{k=0}^{t-1}\pfk,\qquad t\in\T
\] 
is a $Q$-supermartingale for all $X\in\LT$ and all $Q\in\Q^f$, where
\[
\Q^f:=\lk Q\in\ma\mk E_Q\left[\sum_{k=0}^{t}\pfk\right]<\infty\;\,\forall\,t\in\T\rk.
\]
\end{proposition}

The proof of Proposition~\ref{mrc} is a special case of Theorem~\ref{optdec}, which involves the notion of sustainability; cf.\ \cite{ipen7}.

\begin{definition}\label{sustainable}
Let $\rts$ be a dynamic convex risk measure.
We call a bounded adapted process $X=(X_t)\zt$ \emph{sustainable with respect to the risk measure} $\rts$ if
\begin{equation*}
\rt(X_t-X_{t+1})\le0\qquad\textrm{for all $t\in\T$ such that $t<T$}.
\end{equation*}
\end{definition}
Consider $X$ to be a cumulative investment process.  If it is sustainable, then for all $t\in\T$ the adjustment $X_{t+1}-X_t$ is acceptable with respect to $\rt$. 

The next theorem characterizes sustainable processes in terms of a supermartingale inequality; it is a generalization of \cite[Corollary 2.4.10]{ipen7}.

\begin{theorem}\label{optdec}
Let $(\rt)\zt$ be a dynamic convex risk measure such that each $\rt$ is \cts and let $(X_t)\zt$ be a bounded adapted process. Then the following properties are equivalent: 
\begin{enumerate}
\item The process $(X_t)\zt$ is sustainable with respect to the risk measure $(\rt)\zt$.
\item For all  $Q\in\ma$ and all $t\in\T, t\ge 1$, we have
\begin{equation}\label{superm2}
E_Q\left[\,X_{t}\,|\,\F_{t-1}\,\right]\le X_{t-1}+\alpha_{t-1,t}^{\min}(Q)\qquad Q\text{-a.s.}.
\end{equation}
\end{enumerate}
\end{theorem}
\begin{proof} The proof of $1)\ra 2)$ follows directly from the definition of sustainability and the definition of the minimal penalty function.

To prove $2)\ra 1)$, let $(X_t)\zt$ be a bounded adapted process such that \eqref{superm2} holds. In order to prove 
\begin{equation*}
X_t-X_{t-1}=:A_t\in-\A_{t-1,t}\quad\mbox{for all}\quad t\in\T, t\ge 1,
\end{equation*}
suppose by way of contradiction that $A_t\notin-\A_{t-1,t}$. 
Since the set $\A_{t-1,t}$ is convex and weak$^\ast$-closed due to Remark~\ref{abg}, the Hahn-Banach separation theorem (see, e.g., \cite[Theorem A.56 ]{fs4}) ensures the existence of $Z\in L^1(\F_t,P)$ such that
\begin{equation}\label{e15}
a:=\sup_{X\in\A_{t-1,t}}E[\,Z(-X)\,]<E[\,Z\,A_t\,]=:b<\infty.
\end{equation}
Since $\lambda I_{\{Z<0\}}\in\A_{t-1,t}$ for every $\lambda\ge0$,  (\ref{e15}) implies $Z\ge0\;P$-a.s., and in particular $E[Z]>0$. Define a probability measure $Q\in\ma$ via $\frac{dQ}{dP}:=\frac{Z}{E[Z]}$ and note that, due to Lemma \ref{erwpf} and \eqref{e15}, we have
\begin{equation}\label{e172}
E_Q[\alpha_{t-1,t}^{\min}(Q)]=\sup_{X\in\A_{t-1,t}}E_Q[\,(-X)\,]=\sup_{X\in\A_{t-1,t}}E[\,Z(-X)\,]\frac{1}{E[Z]}=\frac{a}{E[Z]}<\infty.
\end{equation}
Moreover, (\ref{e15}) and \eqref{e172} imply
\begin{equation*}
E_Q\left[\left(X_t-X_{t-1}-\alpha_{t-1,t}^{\min}(Q)\right)\right]=E[Z]\left(E[ZA_t]-E_Q\left[\alpha_{t-1,t}^{\min}(Q)\right]\right)= E[Z](b-a)>0,
\end{equation*}
which cannot be true if  \eqref{superm2} holds under $Q$. 
\end{proof}

\begin{remark}
In particular, property 2) of Theorem~\ref{optdec} implies that the process 
\[
X_t-\sum_{k=0}^{t-1}\pfk,\qquad t\in\T
\] 
is a $Q$-supermartingale for all $Q\in\Q^f$, if $X$ is sustainable with respect to $(\rt)$. As shown in \cite[Theorem 2.4.6, Corollary 2.4.8]{ipen7}, this supermartingale property is equivalent to sustainability of $X$ under some additional assumptions. 
\end{remark}

\subsection{Weak time consistency}\label{subsec:wc}
In this section we characterize the weak notions of time consistency from Definition~\ref{cons}. Due to cash invariance, they can be restated as follows: A dynamic convex risk measure $(\rt)\zt$ is weakly acceptance (resp. weakly rejection) consistent, if and only if 
\begin{equation*}
\rtt(X)\le0\quad(\mbox{resp.}\;\ge)\quad\Longrightarrow\quad\rt(X)\le0\quad(\mbox{resp.}\;\ge)
\end{equation*}
for any $X\in\LT$ and for all $t\in\T$ such that $t<T$.
This means that if some position is accepted (or rejected) for any scenario tomorrow, it should be already accepted (or rejected) today. In this form, weak acceptance consistency was introduced in \cite{adehk7}.  Both weak acceptance and weak rejection consistency appeared in \cite{Weber, ros7}. 

Weak acceptance consistency was characterized in terms of acceptance sets in \cite[Corollary 3.6]{tu8}, and in terms of a supermartingale property of penalty functions in \cite[Lemma 3.17]{burg}. We summarize these characterizations in our present setting in the next proposition.

\begin{proposition}\label{weaktc}
Let $(\rt)\zt$ be a dynamic convex risk measure such that each $\rt$ is \ctse. Then the following properties are equivalent: 
\begin{enumerate}
\item $(\rt)\zt$ is weakly acceptance consistent.
\item $\A_{t+1}\subseteq\A_t\quad$ for all $t\in\T$ such that $t<T$.
\item The inequality
\begin{equation}\label{supermart}
E_Q[\,\pft\,|\,F_t\,]\le\pf\quad\qf
\end{equation} 
holds for all $Q\in\ma$ and all $t\in\T$ such that $t<T$. In particular $(\pf)\zt$ is a $Q$-supermartingale for all $Q\in\Q_0$.
\end{enumerate} 
\end{proposition}
\begin{proof} The equivalence of 1) and 2) follows directly from the definition of weak acceptance consistency.  Property 2) implies 3), since by Lemma~\ref{erwpf} 
\begin{align*}
E_Q[\,\pft\,|\,F_t\,]&=\qes_{X_{t+1}\in\At}E_Q[-X_{t+1}|\F_t]\\
&\le\qes_{X\in\A_t}E_Q[-X|\F_t]=\pf\qquad\qf
\end{align*}
for all $Q\in\ma$.

To prove that 3) implies 2), we fix $X\in\At$ and note that
\[
E_Q[-X|\F_{t+1}]\le\pft\quad\qf\quad\mbox{for all}\;\,Q\in\ma
\] 
by the definition of the minimal penalty function. Using (\ref{supermart}) we obtain
\[
E_Q[-X|\F_{t}]\le E_Q[\,\pft\,|\,F_t\,]\le\pf\quad\qf
\]
for all $Q\in\ma$, in particular for $Q\in\Q_t^f(P)$. Thus $\rt(X)\le0$ by \eqref{rd2}.
\end{proof}

\subsection{A recursive construction}\label{recur}

In this section we assume that the time horizon $T$ is finite. Then one can define a time consistent dynamic convex risk measure $({\widetilde \rt})_{t=0\pk T}$ in a recursive way, starting with an arbitrary dynamic convex risk measure $(\rt)_{t=0\pk T}$, via
\begin{equation}\label{recurs}
\begin{aligned}
{\widetilde \rho_T}(X)&:=\rho_T(X)=-X\\
{\widetilde \rho_t}(X)&:=\rho_t(-{\widetilde \rho_{t+1}}(X)),\quad t=0\pk T-1,\quad X\in\LT.
\end{aligned}
\end{equation}
The recursive construction \eqref{recurs} was introduced in \cite[Section 4.2]{cdk6}, and also studied in \cite{Samuel, ck6}. It is easy to see that $({\widetilde \rt})_{t=0\pk T}$ is indeed a time consistent dynamic convex risk measure, and each ${\widetilde \rt}$ is \cts if each $\rt$ has this property.

\begin{remark}\label{cheaper}
If the original dynamic convex risk measure $(\rt)_{t=0\pk T}$ is rejection (resp.\ acceptance) consistent, then the time consistent dynamic convex risk measure $({\widetilde \rt})_{t=0\pk T}$ defined via (\ref{recurs}) lies below (resp. above) $(\rt)_{t=0\pk T}$, i.e.
\[
 {\widetilde \rt}(X)\le(\text{resp.}\ge)\rt(X)\quad\text{for all $t=0\pk T$ and all $X\in\LT$}.
\]
This can be easily proved by backward induction using Proposition~\ref{rejrecursiveness}, monotonicity, and \eqref{recurs}. Moreover, as shown in \cite[Theorem 3.10]{Samuel} in the case of rejection consistency, $({\widetilde \rt})_{t=0\pk T}$ is the biggest time consistent dynamic convex risk measure that lies below $(\rt)_{t=0\pk T}$.
\end{remark}

For all $X\in\LT$, the process $({\widetilde \rt}(X))_{t=0\pk T}$ has the following properties: 
${\widetilde \rho}_T(X)\ge-X$, and
\begin{equation}\label{sust}
\rt({\widetilde \rt}(X)-{\widetilde \rho}_{t+1}(X))=-{\widetilde \rt}(X)+\rt(-{\widetilde \rho}_{t+1}(X))=0\qquad\forall\; t=0\pk T-1,
\end{equation}
by definition and cash invariance. In other words, the process $({\widetilde \rt}(X))_{t=0\pk T}$ covers the final loss $-X$ and is sustainable with respect to the original risk measure $(\rt)_{t=0\pk T}$. The next proposition shows that $({\widetilde \rt}(X))_{t=0\pk T}$ is in fact the smallest process with both these properties. This result is  a generalization of \cite[Proposition 2.5.2 ]{ipen7}, and, in the coherent case, related to \cite[Theorem 6.4]{d6}.

\begin{proposition}\label{smallest}
Let $(\rt)_{t=0\pk T}$ be a dynamic convex risk measure such that each $\rt$ is \ctse. Then, for each $X\in\LT$, the risk process $({\widetilde \rt}(X))_{t=0\pk T}$ defined via (\ref{recurs}) is the smallest bounded adapted process $(U_t)_{t=0\pk T}$ such that $(U_t)_{t=0\pk T}$ is sustainable with respect to $(\rt)_{t=0\pk T}$ and $U_T\ge-X$.
\end{proposition}
\begin{proof} We have already seen that ${\widetilde \rt}_T(X)\ge-X$ and $({\widetilde \rt}(X))_{t=0\pk T}$ is sustainable with respect to $(\rt)_{t=0\pk T}$ due to (\ref{sust}). Now let $(U_t)_{t=0\pk T}$ be another bounded adapted process with both these properties. We will show by backward induction that 
\begin{equation}\label{ineq}
U_t\ge {\widetilde \rt}(X)\quad\f\qquad\forall\;t=0\pk T.
\end{equation}
Indeed, we have
\[
U_T\ge-X={\widetilde \rho_T}(X)\qquad \f.
\]
If (\ref{ineq}) holds for $t+1$, Theorem~\ref{optdec} yields for all $Q\in\Q_t^f$:
\begin{align*}
U_t&\ge E_Q\left[\,U_{t+1}-\pftt\,|\,\F_t\,\right]\nonumber\\&\ge E_Q\left[\,{\widetilde \rho_{t+1}}(X)-\pftt\,|\,\F_t\,\right]\quad \f.
\end{align*}
Thus 
\begin{align*}
U_t&\ge\es_{Q\in\Q_{t}^f} \left(E_Q\left[{\widetilde \rho_{t+1}}(X)|\F_t\right]-\pftt\right)\\
&=\rt(-{\widetilde \rho_{t+1}}(X))={\widetilde \rho_{t}}(X)\quad \f,
\end{align*}
where we have used representation \eqref{rd2}. This proves (\ref{ineq}).\end{proof}
The recursive construction \eqref{recurs} can be used to construct a time consistent dynamic Average Value at Risk, as shown in the next example.
\begin{example}
It is well known that dynamic Average Value at Risk $(AV@R_{t,\lambda_t})_{t=0\pk T}$ (cf.\ Example~\ref{avar}) is not time consistent, and does not even satisfy weaker notions of time consistency from Definition~\ref{cons}; see, e.g., \cite{adehk7, ros7}. Moreover, since $\alpha_0^{\min}(P)=0$ in this case, the set $\Qs$ in \eqref{qstar} is not empty, and \cite[Corollary 4.12]{fp6} implies that there exists no time consistent dynamic convex risk measure $(\rt)\zt$ such that each $\rt$ is continuous from above and $\rho_0=AV@R_{0,\lambda_0}$. However, for $T<\infty$, the recursive construction \eqref{recurs} can be applied to $(AV@R_{t,\lambda_t})_{t=0\pk T}$ in order to modify it to a time consistent dynamic coherent risk measure $(\tilde{\rt})_{t=0\pk T}$. This modified risk measure takes the form
\begin{align*}
\tilde{\rt}(X)&= \es\lk E_Q[-X|\F_t]\mk Q\in\Q_t, \frac{Z^Q_{s+1}}{Z^Q_s}\leq \lambda_s^{-1}, s=t\pk T-1\rk \\
&= \es\lk E_P\left[-X\prod_{s=t+1}^{T}L_s\mk\F_t\right]\mk L_s\in L^{\infty}_s, 0\leq L_s\leq \lambda_s^{-1}, E[L_s|\F_{s-1}]=1, s=t+1\pk T\rk\nonumber
\end{align*}
for all $t=0\pk T-1$, where $Z^Q_t=\frac{dQ}{dP}|_{\F_t}$. This was shown, e.g., in \cite[Example 3.3.1]{ck6}.
\end{example}

\section{The dynamic entropic risk measure}\label{entropic}
In this section we study time consistency properties of the dynamic entropic risk measure
\[
\rt(X)=\frac{1}{\g_t}\log E[e^{-\g_t X}|\F_t],\qquad t\in\T,\qquad X\in\LT,
\]
where the risk aversion parameter $\gamma_t$ satisfies $\gamma_t>0\, P$-a.s. and $\gamma_t, \frac{1}{\gamma_t}\in\Lt$ for all $t\in\T$ (cf.\ Example~\ref{ex:entr}).

It is well known (see, e.g., \cite{dt5, fp6}) that the conditional entropic risk measure $\rho_{t}$ has the robust representation \eqref{rd1} with the minimal penalty function $\alpha_t$ given by
\[
\alpha_t(Q)=\frac{1}{\gamma_t}H_t(Q|P),\quad Q\in\Q_t,
\]
where $H_t(Q|P)$ denotes the conditional relative entropy of $Q$ with respect to $P$ at time $t$:
\[
H_t(Q|P)=E_Q\lp \log\frac{dQ}{dP}\mk\F_t \rp,\quad Q\in\Q_t.
\] 

The dynamic entropic risk measure with constant risk aversion parameter $\gamma_t=\gamma_0\in\real$ for all $t$ was studied in \cite{dt5,cdk6,fp6,ck6}. It plays a particular role since, as proved in \cite{ks9}, it is the only law invariant time consistent relevant dynamic convex risk measure.

In this section we consider an \emph{adapted} risk aversion process $(\gamma_t)\zt$, that depends both on time and on the available information. As shown in the next proposition, the process $(\gamma_t)\zt$ determines time consistency properties of the corresponding dynamic entropic risk measure. This result corresponds to \cite[Proposition 4.1.4]{ipen7}, and generalizes \cite[Proposition 3.13]{Samuel}. 
\begin{proposition}\label{riskav}
Let $(\rt)\zt$ be the dynamic entropic risk measure with risk aversion given by an adapted process $(\g_t)\zt$ such that $\gamma_t>0\,P$-a.s. and $\gamma_t, 1/\g_t\in\Lt$. Then the following assertions hold:
\begin{enumerate}
\item $(\rt)\zt$ is rejection consistent if $\g_t\ge\g_{t+1}\:P$-a.s. for all $t\in\T$ such that $t<T$;
\item $(\rt)\zt$ is acceptance consistent if $\g_t\le\g_{t+1}\:P$-a.s. for all $t\in\T$ such that $t<T$;
\item $(\rt)\zt$ is time consistent if $\g_t=\g_0\in\real\: P$-a.s. for all $t\in\T$ such that $t<T$.
\end{enumerate}
Moreover, assertions 1), 2) and 3) hold with ``if and only if'', if $\g_t\in\real$ for all $t$, or if the filtration $(\F_t)\zt$ is rich enough in the sense that for all $t$ and for all $B\in\F_t$ such that $P[B]>0$ there exists $A\subset B$ such that $A\notin\F_t$ and $P[A]>0$.
\end{proposition}

\begin{proof} Fix $t\in\T$ and $X\in\LT$. Then
\begin{align*}
\rt(-\rtt(X))&=\frac{1}{\g_t}\log\left(E\left[\exp\left\{\frac{\g_t}{\g_{t+1}}\log\left(E\left[e^{-\g_{t+1}X}|\F_{t+1}\right]\right)\right\}\big|\F_t\right]\right)\nonumber\\
&=\frac{1}{\g_t}\log\left(E\left[E\left[e^{-\g_{t+1}X}|\F_{t+1}\right]^{\frac{\g_t}{\g_{t+1}}}\big|\F_t\right]\right).
\end{align*}
Thus $\rt(-\rtt)=\rt$ if $\g_t=\g_{t+1}$ and this proves time consistency. Rejection (resp. acceptance) consistency follow by the generalized Jensen inequality that will be proved in Lemma~\ref{jensen}. We apply this inequality at time $t+1$ to the bounded random variable $Y:=e^{-\g_{t+1}X}$ and the $\B\left((0,\infty)\right)\otimes\F_{t+1}$-measurable function 
\[
u\;:\;(0,\infty)\times\Omega\;\rightarrow\;\real,\quad\quad u(x,\omega):=x^{\frac{\g_t(\omega)}{\g_{t+1}(\omega)}}.
\]
Note that $u(\cdot,\omega)$ is convex if $\g_t(\omega)\ge\g_{t+1}(\omega)$ and concave if $\g_t(\omega)\le\g_{t+1}(\omega)$. Moreover, $u(X,\cdot)\in\LT$ for all $X\in\LT$ and $u(\cdot,\omega)$ is differentiable on $(0,\infty)$ with 
\[
|u'(x,\cdot)|=\frac{\g_t}{\g_{t+1}}x^{\frac{\g_t}{\g_{t+1}}-1}\le ax^b\quad\f
\]
for some $a,b\in\real$ if $\g_t\ge\g_{t+1}$, due to our assumption $\frac{\g_t}{\g_{t+1}}\in\LT$. On the other hand, for $\g_t\le\g_{t+1}$ we obtain
\[
|u'(x,\cdot)|=\frac{\g_t}{\g_{t+1}}x^{\frac{\g_t}{\g_{t+1}}-1}\le a\frac{1}{x^c}\quad\f
\]
for some $a,c\in\real$. Thus the assumptions of Lemma~\ref{jensen} are satisfied and we obtain 
\[
\rt(-\rtt)\le\rt\quad\mbox{if}\quad\g_t\ge\g_{t+1}\quad P\mbox{-a.s.\;\; for all $t\in\T$ such that $t<T$}
\]
and 
\[
\rt(-\rtt)\ge\rt\quad\mbox{if}\quad\g_t\le\g_{t+1}\quad P\mbox{-a.s.\;\; for all $t\in\T$ such that $t<T$}.
\]
The ``only if'' direction for constant $\g_t$ follows by the classical Jensen inequality.

Now we assume that the sequence $(\rt)\zt$ is  rejection consistent and our assumption on the filtration $(\F_t)\zt$ holds. We will show that the sequence $(\g_t)\zt$ is decreasing in this case. Indeed, for $t\in\T$ such that $t<T$, consider $B:=\{\gamma_t<\gamma_{t+1}\}$ and suppose that $P\left[ B \right]>0$.
Our assumption on the filtration allows us to choose $A \subset B$ with $P\left[ B \right]>P\left[ A \right]>0$ and $A\notin\F_{t+1}$. We define a random variable $X:=-xI_{A}$ for some $x>0$. Then
\begin{align*}
\rho_{t}( -\rho_{t+1}(X))&=\frac{1}{\gamma_t}\log \left( E\left[ \exp\left( \frac{\gamma_t}{\gamma_{t+1}}\log \left( E\left[ e^{ \gamma_{t+1}x I_{A}}\big |\F_{t+1} \right] \right) \right) \big|\F_t \right] \right)\\
&=\frac{1}{\gamma_t}\log \left( E\left[ \exp\left( \frac{\gamma_t}{\gamma_{t+1}} I_{B}\log \left( E\left[ e^{ \gamma_{t+1}x I_{A}}\big |\F_{t+1} \right] \right) \right) \big|\F_t \right] \right),
\end{align*}
where we have used that $A\subset B$. Setting 
\begin{equation*}
Y:=E\left[ e^{ \gamma_{t+1}x I_{A}}\big |\F_{t+1} \right]\\
=e^{ \gamma_{t+1}x }P\left[ A |\F_{t+1} \right]+P\left[ A^c |\F_{t+1} \right]
\end{equation*}
and bringing $\frac{\g_t}{\g_{t+1}}$ inside of the logarithm we obtain
\begin{equation}\label{a1}
\rho_{t}\left( -\rho_{t+1}\left( X \right) \right)=\frac{1}{\gamma_t}\log \left( E\left[ \exp\left( I_{B}\log \left( Y^{\frac{\gamma_t}{\gamma_{t+1}} I_{B}} \right) \right) \big|\F_t \right] \right).
\end{equation}
The function $x \mapsto x^{\gamma_{t}(\omega)/\gamma_{t+1}(\omega)}$ is strictly concave for almost each $\omega\in B$, and thus 
\begin{align}\label{a2}
Y^{\frac{\gamma_t}{\gamma_{t+1}} }&=\left(e^{\gamma_{t+1}x }P\left[ A |\F_{t+1} \right]+(1-P\left[ A |\F_{t+1} \right])\right)^{\frac{\gamma_t}{\gamma_{t+1}}}\nonumber\\
&\ge e^{ \gamma_{t}x }P\left[ A |\F_{t+1}\right]+(1-P\left[ A |\F_{t+1} \right])\qquad P\mbox{-a.s on}\;B,
\end{align}
with strict inequality on the set 
\[
C:=\left\{P\left[ A |\F_{t+1}\right]>0\right\}\cap \left\{P\left[ A |\F_{t+1}\right]<1\right\}\cap B.
\]
Our assumptions $P[A]>0, \,A\subset B$ and $A\notin\F_{t+1}$ imply $P[C]>0$ and using 
\begin{equation}\label{a3}
e^{ \gamma_{t}x }P\left[ A |\F_{t+1}\right]+(1-P\left[ A |\F_{t+1} \right])=E\left[e^{\g_txI_A}|\F_{t+1}\right]
\end{equation}
we obtain from (\ref{a1}), (\ref{a2}) and (\ref{a3})
\begin{equation}\label{0}
\rho_{t}\left( -\rho_{t+1}\left( X \right) \right)\ge\frac{1}{\gamma_t}\log \left( E\left[ \exp\left( I_{B}\log \left(E\left[e^{\gamma_{t}x I_{A}}|\F_{t+1} \right] \right) \right) \big|\F_t \right] \right),
\end{equation}
with the strict inequality on some set of positive probability due to strict monotonicity of the exponential and the logarithmic functions. For the right hand side of (\ref{0}) we have
\begin{align*}
\lefteqn{\frac{1}{\gamma_t}\log \left( E\left[ \exp\left( I_{B}\log \left(E\left[e^{\gamma_{t}x I_{A}}|\F_{t+1} \right] \right) \right) \big|\F_t \right] \right)=}\\
&=\frac{1}{\gamma_t}\log \left( E\left[ I_{B}E\left[e^{ \gamma_{t}x I_{A} }|\F_{t+1} \right]+I_{B^c} \big|\F_t \right] \right)\\
&=\frac{1}{\gamma_t}\log \left( E\left[ \exp\left( \gamma_{t}x I_{A} \right)\big|\F_t \right] \right)\\
&=\rho_{t}\left( X \right),
\end{align*}
where we have used $A\subset B$ and $B\in\F_{t+1}$. This is a contradiction to  rejection consistency of $(\rt)\zt$, and we conclude that $\g_{t+1}\le\g_t$ for all $t$.  The proof in the case of  acceptance consistency follows in the same manner. And since time consistent dynamic risk measure is both  acceptance and  rejection consistent, we obtain $\g_{t+1}=\g_t$ for all $t$.
\end{proof}

The following lemma concludes the proof of Proposition~\ref{riskav}.

\begin{lemma}\label{jensen}
Let $(\Omega,\F, P)$ be a probability space and $\F_t\subseteq \F$ a $\sigma$-field. Let $I\subseteq\real$  be an open interval and
\[
u\;:\;I\times\Omega\;\rightarrow\;\real
\]
be a $\B\left(I\right)\otimes\F_{t}$-measurable function such that $u(\cdot, \omega)$ is convex (resp. concave) and finite on $I$ for $P$-a.e. $\omega$. Assume further that 
\[
|u_+'(x,\cdot)|\le c(x)\quad P\mbox{-a.s. with some}\;\;c(x)\in\real\;\,\mbox{for all}\;\;x\in I,
\]
where $u_+'(\cdot,\omega)$ denotes the right-hand derivative of $u(\cdot,\omega)$. Let $X\;:\;\Omega\;\rightarrow\;[a,b]\subseteq I$ be an $\F$-measurable bounded random variable such that $E\left[\,|u(X,\, )|\,\right]<\infty$. Then
\[
E\left[\,u(X,\, )\,|\,\F_t\,\right]\ge u\left(E[X|\F_t],\, \right)\quad(\mbox{resp}\le)\quad\f.
\]
\end{lemma} \begin{proof} We will prove the assertion for the convex case; the concave one follows in the same manner. Fix $\omega\in\Omega$ such that $u(\cdot, \omega)$ is convex. Due to convexity we obtain for all $x_0\in I$
\[
u(x,\omega)\ge u(x_0,\omega)+u_+'(x_0,\omega)(x-x_0)\quad \mbox{for all}\quad x\in I.
\]
Take $x_0=E[X|\F_t](\omega)$ and $x=X(\omega)$. Then
\begin{equation}\label{dk}
u(X(\omega),\omega)\ge u(E[X|\F_t](\omega),\omega)+u_+'(E[X|\F_t](\omega),\omega)(X(\omega)-E[X|\F_t](\omega))
\end{equation}
for $P$-almost all $\omega\in\Omega$. Note further that $\B\left(I\right)\otimes\F_{t}$-measurability of $u$ implies $\B\left(I\right)\otimes\F_{t}$-measurability of $u_+$. Thus
\[
\omega\,\rightarrow\, u(E[X|\F_t](\omega),\omega)\quad\mbox{and}\quad \omega\,\rightarrow\,u_+'(E[X|\F_t](\omega),\omega)
\] 
are $\F_t$-measurable random variables, and $\omega\rightarrow u(X(\omega),\omega)$ is $\F$-measurable. Moreover, due to our assumption on $X$, there are constants $a,b\in I$ such that $a\le E[X|\F_t]\le b\;P$-a.s.. Since $u_+'(\cdot,\omega)$ is increasing by convexity, by using our assumption on the boundedness of $u_+'$ we obtain
\[
-c(a)\le u_+'(a,\omega)\le u_+'(E[X|\F_t],\omega)\le u_+'(b,\omega)\le c(b),
\]
i.e. $u_+'(E[X|\F_t],\,)$ is bounded. Since $E\left[\,|u(X,\, )|\,\right]<\infty$, we can build conditional expectation on the both sides of (\ref{dk}) and we obtain
\begin{align*}
E[\,u(X,\,)\,|\,\F_t\,]&\ge E\left[\,u(E[X|\F_t],\,)+u_+'(E[X|\F_t],\,)(X-E[X|\F_t])\,|\,\F_t\,\right]\\
&=E\left[\,u(E[X|\F_t],\,)\,|\,\F_t\,\right]\quad\f,
\end{align*}
where we have used $\F_t$-measurability of $u(E[X|\F_t],\,)$ and  of $u_+'(E[X|\F_t],\,)$ and the boundedness of $u_+'(E[X|\F_t],\,)$. This proves our claim.\end{proof}

\bibliographystyle{plain}

\end{document}